\documentclass[12pt]{article}
\usepackage{amsmath, amssymb, amsfonts, amsthm}
\usepackage{color}
\title{Stability and semiclassics in self-generated fields}
\author{L\'aszl\'o Erd\H os
 \thanks{Partially supported by SFB-TR12 of
the German Science Foundation. {\text lerdos@math.lmu.de} }
\\Institute of Mathematics, University of Munich \\
Theresienstr. 39, D-80333 Munich, Germany \\
S\o ren Fournais \thanks{Work partially supported by the Lundbeck
  Foundation, the Danish Natural Science Research Council and the European 
Research Council under the
 European Community's Seventh Framework Program (FP7/2007--2013)/ERC grant
 agreement  202859.
{\text fournais@imf.au.dk}} \\ Department of Mathematical Sciences, Aarhus University\\
 Ny Munkegade 118, DK-8000 Aarhus, Denmark
\\ and \\
Jan Philip Solovej \thanks{Work partially supported
   by the Danish Natural Science Research Council and by a Mercator
   Guest Professorship from the German Science Foundation. {\text
solovej@math.ku.dk}}
\\ Department of Mathematics, University of Copenhagen\\
Universitetsparken 5, DK-2100 Copenhagen,
Denmark}

\date{Oct 18, 2011}

\newtheorem{theorem}{Theorem}[section]
\newtheorem{proposition}[theorem]{Proposition}

\newcommand{\rd}{{\rm d}}
\newcommand{\be}{\begin{equation}}
\newcommand{\ee}{\end{equation}}
\newcommand{\bey}{\begin{align}}
\newcommand{\eey}{\end{align}}
\newcommand{\beys}{\begin{align*}}
\newcommand{\eeys}{\end{align*}}

\newcommand{\bZ}{{\mathbb Z}}

\newcommand{\bsigma}{\mbox{\boldmath $\sigma$}}

\renewcommand{\iint}{\int \!\! \int}

\newcommand{\bR}{{\mathbb R}}
\newcommand{\bS}{{\mathbb S}}

\newcommand{\bC}{{\mathbb C}}

\newcommand{\bN}{{\mathbb N}}

\newcommand{\e}{\varepsilon}

\newcommand{\Tr}{{\rm Tr\;}}
\newcommand{\tr}{{\rm Tr\;}}

\newcommand{\wh}{\widehat}
\newcommand{\wt}{\widetilde}

\newcommand{\cF}{{\cal F}}

\newcommand{\cE}{{\cal E}}

\newcommand{\cD}{{\cal D}}

\newcommand{\cN}{{\cal N}}

\newcommand{\al}{\alpha}

\newcommand{\non}{\nonumber}

\begin{document}
\maketitle

\begin{abstract} 
  We consider non-interacting particles subject to a fixed external
  potential $V$ and a self-generated magnetic field $B$.  The total
  energy includes the field energy $\beta \int B^2$ and we minimize
  over all particle states and magnetic fields.  In the case of
  spin-$1/2$ particles this minimization leads to the coupled
  Maxwell-Pauli system.  The parameter $\beta$ tunes the coupling
  strength between the field and the particles and it effectively
  determines the strength of the field.  We investigate the stability
  and the semiclassical asymptotics, $h\to0$, of the total ground
  state energy $E(\beta, h, V)$. The relevant parameter measuring the
  field strength in the semiclassical limit is $\kappa=\beta h$.  We
  are not able to give the exact leading order semiclassical
  asymptotics uniformly in $\kappa$ or even for fixed $\kappa$. We do
  however give upper and lower bounds on $E$ with almost matching
  dependence on $\kappa$. In the simultaneous limit $h\to0$ and
  $\kappa\to\infty$ we show that the standard non-magnetic Weyl
  asymptotics holds. The same result also holds for the spinless case,
  i.e. where the Pauli operator is replaced by the Schr\"odinger
  operator.
\end{abstract}



\medskip\noindent
{\it Running title:} Stability and semiclassics with self-generated field.

\section{Introduction}\label{sec:intro}

\bigskip

An important problem in spectral analysis is to determine or bound the sum of
the negative eigenvalues of a Schr\"odinger operator
$-\Delta-V(x)$, i.e., 
$$
\tr(-\Delta-V(x))_-
$$
under appropriate conditions on the potential $V$. We use the
convention that $x_-= (x)_- =\min\{ x, 0\}$ when $x$ is either a real
number or a self-adjoint operator.  This problem is of
particular interest in quantum mechanics as it gives 
 the ground state energy of a gas
of free fermions moving in the exterior potential $V$.

A generalization of this problem is to consider not only a potential
$V$ but also an exterior magnetic field given by the vector potential
$A$. The corresponding magnetic Schr\"odinger operator is 
$
(-i\nabla + A)^2-V(x)
$ acting in $L^2(\bR^d)$.
A further generalization is to consider the particles as having 
spin$-\frac{1}{2}$
and introduce the magnetic Pauli operator 
$[\bsigma\cdot (-i\nabla + A)]^2-V$, where in  $d=3$
dimensions $\bsigma=(\sigma_1, \sigma_2, \sigma_3)$ denotes the vector
of $2\times2$ Pauli matrices. The Pauli operator acts in $L^2(\bR^3;\bC^2)$.
Much work has gone into understanding the semiclassical asymptotics of
the sum of negative eigenvalues, i.e., the asymptotics for small $h>0$
of
$$
\tr((-ih\nabla + A)^2-V(x))_-\quad
\text{or}\quad \tr([\bsigma\cdot (-ih\nabla + A)]^2-V(x))_-.
$$
It is well known that under appropriate conditions on $A$ and $V$ the
leading behavior as $h$ tends to zero is given by the Weyl formulas
$$
(2\pi h)^{-d}\iint_{\bR^d\times\bR^d} (p^2-V(x))_-\rd x\rd p \quad
\text{or}\quad 2(2\pi h)^{-3}\iint_{\bR^3\times\bR^3} (p^2-V(x))_-\rd x\rd p,
$$
respectively. Here the factor of 2 on the second integral is due to
the spin degrees of freedom in the Pauli operator, i.e., the fact that
it is a $2\times2$-matrix valued operator.  Note that the limiting
semiclassical behavior is non-magnetic, i.e.  fixed magnetic fields do
not influence the leading order semiclassics.  For simplicity we will
consider the $d=3$ dimensional case only and we denote the
Schr\"odinger operator $T_h^{\rm S}(A)=(-ih\nabla + A)^2$ and the Pauli operator
$T_h^{\rm P}(A)=[\bsigma\cdot (-ih\nabla + A)]^2$. The magnetic field is
$B=\nabla\times A$.

In this paper we will address a related and equally important issue,
namely the case when the magnetic field is not a fixed external field,
but the self-generated classical magnetic field generated by the
particles themselves. 
 
We will consider the external potential $V$ to be a fixed (given)
function in $\bR^3$, we assume $V\in L_{loc}^1(\bR^3)$ 
and we will always work with the flat Euclidean
metric. The vector potential $A$ will be optimized to
minimize the total energy consisting of the energy of the particles
and the field energy
$$
\int B^2 =\int |\nabla\times A|^2
$$
(we use the convention that unspecified integrals are always
on $\bR^3$ w.r.t. the Lebesgue measure). 
The problem we consider is thus to determine the
energy
\begin{equation}\label{eq:energy}
E^{\rm S,P}(\beta,h,V)=\inf_A\left[\tr(T^{\rm S,P}_h(A)-V)_-+\beta\int |\nabla\times A|^2\right]
\end{equation}
for $\beta, h>0$, where the infimum runs over all vector fields $A\in
H^1(\bR^3;\bR^3)$; in fact minimizing only for all $A\in
C_0^\infty(\bR^3;\bR^3)$ gives the same infimum. See
Appendix~\ref{vectorfields} for a discussion of equivalent variational
spaces for this energy and for the precise definition of the operator
$T^{\rm S,P}_h(A)-V$ and the sum of its negative eigenvalues. We will
omit the superscripts S,P, when making general statements valid for
both the Schr\"odinger and Pauli cases.

Here $\beta$ is an additional parameter setting the strength of the
coupling of the particles to the field. In a given physical system the
values of $h$ and $\beta$ are given, but as is standard in
semiclassical analysis we leave them as free parameters. Formally
$\beta=\infty$ corresponds to the non-magnetic case; smaller $\beta$
means that a larger effect of the magnetic field is expected.

The Euler-Lagrange  equation corresponding to the variational problem
\eqref{eq:energy} above is 
\be
\beta\,\nabla\times B=J_A,
\label{maxwell}
\ee
where $J_A$ is the current of the Fermi gas, which in the
Schr\"odinger case is
$$
J_A(x)=-\text{Re}\,\left[(-ih\nabla+A)1_{(-\infty,0]}(T_h(A)-V)\right](x,x)
$$ 
and in the Pauli case is
$$
J_A(x)=-\text{Re}\,\left[\tr_{\bC^2}\left(\bsigma (\bsigma\cdot(-ih\nabla+A))1_{(-\infty,0]}(T_h(A)-V)\right)\right](x,x).
$$
In other words the Euler-Lagrange equations are the non-linear coupled
(time independent) Maxwell-Schr\"odinger or Maxwell-Pauli equations
where we deliberately ignored Gauss' equation for $V$ in order
to obtain a general result. In the application to large atoms discussed
below we will consider the special case when $V$ 
solves Gauss' equation.

There are three natural questions about the energy $E(\beta,h,V)$. First
of all we may ask whether the energy is finite, i.e., not negative
infinity. We refer to this as {\it stability}. In
Theorem~\ref{thm:stab} we give bounds on the energy in the Pauli case under 
essentially sharp assumptions on $V$. The corresponding results for the
Schr\"odinger case are well known and are also discussed in
Section~\ref{sec:results}. 

The second natural question is whether the
inclusion of the self-generated magnetic field will actually lower the
energy at all. For both the Pauli and the Schr\"odinger cases this is indeed 
true (see Appendix~\ref{sec:paramagnetism}), i.e., there exist
potentials $V$ and parameters $\beta,h$ (not necessarily the same for
the Pauli and Schr\"odinger cases) such that
\begin{equation}\label{eq:paramagnetism}
E(\beta,h,V)<E(\beta=\infty,h,V).
\end{equation}

The third and, actually, main question we address is how the inclusion
of the magnetic field influences the semiclassical asymptotics ($h\to
0$) for the sum of the negative eigenvalues.  Standard semiclassical
results typically assume that the physical data (external potential,
vector potential, metric etc.) are smooth. If as above these data
arise as self-generated and thus determined internally via a
variational principle the smoothness is not a-priori given.

In Theorem~\ref{thm:weak} we establish under the appropriate assumptions
on $V$ that the semiclassical, i.e., $h\to0$, asymptotics
of $E(\beta,h,V)$ in the case when $\beta h\to\infty$ simultaneously
with $h\to0$, is given by the standard non-magnetic Weyl formula. 
The case of large $\beta h$ is the case of greatest physical interest
as the magnetic field in general gives rise to a small effect. 
In \cite{EFS2} we give improved estimates on the error term
to the non-magnetic Weyl term.

It is however at least of mathematical interest to understand the
behavior of $E(\beta,h,V)$ also when $\beta h$ is smaller, i.e., when
the effect of the magnetic field is greater.  Unfortunately, we have
not been able to establish the exact semiclassical asymptotics in this
case.  We do show however in Theorem~\ref{thm:stab} again under
appropriate assumptions on $V$ that the Pauli energy $h^3E^{\rm
  P}(\beta,h,V)$ is bounded from above and below by functions of
$\kappa:=\beta h$.  These functions have almost matching asymptotics
in the regime of small $\beta h$. In proving the upper bound we rely
heavily on the construction of zero-modes in \cite{ES1}.  The similar
result for the Schr\"odinger case is well-known and, in fact, here the
bounds do not depend on $\kappa$ (see (\ref{schrstab}) below).

Theorem~\ref{thm:stab}, moreover
shows that for sufficiently small $\kappa = \beta h$ (depending on
$V$) the magnetic and non-magnetic energies are different, i.e.,
(\ref{eq:paramagnetism}) holds, since the former is independent of
$\kappa$, while the latter scales at least as $\kappa^{-3+\e}$ for any
$\e>0$.  In summary, we have established that $\beta \sim h^{-1}$ is
the correct threshold for the Pauli operator to observe the effect of
the magnetic field in the leading order semiclassics.

Our results will be stated and proved for the Pauli operator
and we will remark the modifications for the Schr\"odinger case.

\subsection{Applications to large atoms}

One of the main applications of precise semiclassical estimates is to 
investigate the ground state energy of large atoms and molecules.
It is a celebrated result of Lieb and Simon \cite{LS} (see also \cite{L}) that the energy
of a neutral atom or molecule with nuclear charge $Z$ obeys 
the Thomas-Fermi asymptotics, $-(const.)Z^{7/3}$, in the large
$Z$ limit.    The subleading correction to order $Z^2$ is known as the
Scott correction and was established for atoms 
in \cite{H,SW1} and for molecules
in \cite{IS} (see also \cite{SW2,SW3,SS}).
 The next term in the
expansion of order $Z^{5/3}$ was rigorously established for atoms in
\cite{FS}.  The large $Z$
asymptotics can be viewed as a semiclassical limit with $h=Z^{-1/3}$
being the semiclassical parameter. The Scott term thus corresponds to
the next order semiclassical estimate. In the results mentioned so far
on the large $Z$ asymptotics, magnetic interactions are ignored. 

Magnetic fields in this context were first
taken into account  as  external fields, either
a homogeneous one \cite{Y,LSY1, LSY2} (see also \cite{Sob1,Sob2,Sob3} for  improved
semiclassical estimates and \cite{Iv1,Iv2} for inclusion of the Scott correction)
or an inhomogeneous one \cite{ES1} but subject to certain
regularity conditions. Self-generated magnetic fields, obtained
from Maxwell's equation \eqref{maxwell} are not known to satisfy these conditions.
In \cite{ES3} the validity of
Thomas-Fermi theory was extended by allowing a self-generated magnetic field
that interacts with the electrons. This means that the focus was on 
 the absolute ground state of the system, after minimizing
for both the electron wave function and for the magnetic
field. Without going into details we mention that the Thomas-Fermi theory
can be viewed as the semiclassical approximation to 
our Maxwell-Pauli system but with the Gauss' equation included.
It was shown  that the additional magnetic field does
not change the leading order Thomas-Fermi energy. This
holds if $Z\alpha^2$ is sufficiently small, where $\alpha$ is
the fine structure constant; for large values of $Z\alpha^2$ 
the system is unstable (see \cite{ES3} for more details).

The semiclassical problem corresponding to a self-generated magnetic
field is exactly of the type \eqref{eq:energy} we discuss in the
present paper for $\beta h\to\infty$.  In order to establish the Scott
correction for self-generated magnetic fields, it is necessary to
establish the semiclassical expansion of $E(\beta, h, V)$ up to
subleading order if $\beta h^2$ is bounded from below. Such an
improved semiclassical estimate is proven in a separate paper
\cite{EFS2} and the Scott term asymptotics is proved in \cite{EFS3}.
\bigskip

{\bf Acknowledgment:} JPS thanks R.\ Seiringer and C.\ Hainzl for
fruitful discussions.
\section{Results}\label{sec:results}

In this section we will state and discuss our two main theorems. The
proofs are given in the following sections. 
\begin{theorem}[Stability bounds for the Pauli energy]\label{thm:stab}
Assuming  $V\in L^1_{loc}(\bR^3)$ and $[V]_+\in L^{5/2}(\bR^3)\cap L^4(\bR^3)$,  we have, for the
Pauli operator
\be
  h^3 E^{\rm P}(\beta, h, V)\geq -C \int [V]_+^{5/2} -C (\beta h)^{-3} \int [V]_+^4\label{paulistablow}.
\ee
for $C>0$ and all $h\in (0,\infty)$ and $\beta
\in(0, \infty]$. On the other hand if 
$V\in C_0^1(\bR^3)$ then for all $0<\varepsilon<1/3$, $h\in (0,\infty)$ and $\beta \in(0, \infty]$
\begin{equation}
h^3E^{\rm P}(\beta, h, V) \leq- C' \int [V]_+^{5/2} 
-C_\varepsilon(\beta h)^{-3+2\varepsilon}\int [V]_+^{4-\varepsilon} + \cE_V(h)\label{paulistabup}
\end{equation}
for positive constants $C'$ and  $C_\varepsilon$, the latter depending on $\varepsilon$,
and with an error function $h\mapsto \cE_V(h)$, depending only on
$\varepsilon$, $h$,and $V$
and satisfying $\lim_{h\to0}\cE_V(h) =0$.
\end{theorem}
For  reference, we mention the analogous result for the Schr\"odinger case. 
If $V\in L^{5/2}(\bR^3)$, we have for all $h>0$ and $\beta \in[0, \infty]$
\be
   -C\int [V]_+^{5/2} \le h^3 E^{\rm S}(\beta, h, V)  \le -\frac{1}{15\pi^2}\int [V]_+^{5/2}
+\cE_V(h)
\label{schrstab}
\ee 
with some positive constant $C$ and with an error function
depending on $V$ and $h$ and again satisfying $\lim_{h\to0}\cE_V(h)
=0$. The lower bound in (\ref{schrstab}) is the classical
Lieb-Thirring inequality \cite{LT}, which holds also for magnetic
Schr\"odinger operators (see e.g., \cite{S}) . The upper bound
is achieved by setting $A=0$ and using the standard Weyl semiclassical
estimate for the non-magnetic operator. 

Not only does 
Theorem~\ref{thm:stab} give the almost sharp asymptotics of $h^3E^{\rm
  P}(\beta,h,V)$ in the limit of small $\kappa=\beta h$, it also essentially gives the 
optimal condition on which $L^p$-norms of $V$ are needed to bound the energy 
$h^3E^{\rm P}(\kappa h^{-1},h,V)$ uniformly in $h$ for small $h$
and fixed  $\kappa$. {F}rom the
lower bound in the theorem the $L^{5/2}$ and $L^4$ norms bound the energy.
Conversely, from the upper bound all $L^p$ norms with $p\in[5/2,4)$
are needed since
if $V_n\in C_0^1$ is a sequence such that
$\|[V_n]_+\|_p\to\infty $ as $n\to\infty$ for some $p\in[5/2,4)$, there is a sequence $h_n$
tending to zero such that $h_n^3E^{\rm P}(\kappa h_n^{-1},h_n,V)\to-\infty$.

{\it Remark 1.} 
We can also ask a slightly different question: What $L^p$-norm  {\it condition} 
on  $V$ will ensure finiteness of the energy, but not necessarily a
{\it bound} in terms of this norm?
 For the Schr\"odinger operator the answer is that $V_+\in
L^{3/2}(\bR^3)$ ensures finiteness of $E^{\rm S}(\beta,h,V)$ and
this is essentially sharp as far as the local singularity
is concerned, since for the critical case
$V(x)=c|x|^{-2}$ the
energy is finite if $c\leq h^2/4$ and infinite otherwise. 
For the Pauli case we do not know a similar sharp result. 

If we ask instead for finiteness of only one eigenvalue, i.e., the
one-electron energy
\begin{equation}
E_0^{\rm S,P}(\beta,h,V)=\inf_A\left[\inf\text{ Spec}(T^{\rm S,P}_h(A)-V)
+\beta\int |\nabla\times A|^2\right]
\end{equation}
we have a sharp result for both Pauli and Schr\"odinger. Of course 
$E_0(\beta,h,V)\geq E(\beta,h,V)$ and hence instability for $E_0$
implies instability for $E$. The situation
for Schr\"odinger is exactly the same for $E^{\rm S}_0$ and $E^{\rm S}$, i.e.
if $V_+\in
L^{3/2}(\bR^3)$ then $E_0^{\rm S}$ is finite and for the critical case $V(x)=c|x|^{-2}$ the
energy is finite if $c\leq h^2/4$ and infinite otherwise.
 
For Pauli $E_0^{\rm p}$ is finite if $V_+\in L^3(\bR^3)$ and for the
critical case $V(x)=c|x|^{-1}$ there exists a critical value
$\gamma_{\rm cr}>0$ such that $E_0^{\rm p}$ is finite if
$c<\gamma_{\rm cr}\beta h^2$ and infinite if $c>\gamma_{\rm cr}\beta
h^2$. This is essentially contained in \cite{FLL} and we shall review
it briefly in Appendix~\ref{app:stab}.  We will also discuss in the
appendix that the sum of all eigenvalues $E^{\rm p}$ likewise remains
bounded for the cutoff Coulomb potential $V(x)=[c|x|^{-1}-1]_+$ if
$c>0$ is small enough.

Finally, we give the result of the exact spectral asymptotics in the
case of weak magnetic fields. 

\begin{theorem}[Semiclassics for weak fields]\label{thm:weak} 
Assume that $V\in L^{5/2}(\bR^3)\cap L^4(\bR^3)$ then
\be\label{eq:asymppauli}
\lim_{h\to0\atop \beta h\to \infty}  h^3 E^{\rm P}(\beta, h, V) = 
 \lim_{h\to0}  h^3 E^{\rm P}(\infty, h, V) = -\frac{2}{15\pi^2}\int [V(x)]_+^{5/2}\rd x.
\ee
Likewise for the Schr\"odinger case we have for $V\in L^{5/2}(\bR^3)$ 
\be\label{eq:asymp}
\lim_{h\to0\atop \beta h\to \infty}  h^3 E^{\rm S}(\beta, h, V) =  
\lim_{h\to0}  h^3 E^{\rm S}(\infty, h, V) = -\frac{1}{15\pi^2}\int [V(x)]_+^{5/2}\rd x.
\ee
\end{theorem}
{\it Remark 2.} This result is
a strengthening of Theorem 1.3 from \cite{ES3}, where the
same conclusion was proved under the condition $\beta h^2\ge c>0$.

{\it Remark 3.}
We conjecture that in case of the Schr\"odinger operator the Weyl term is
the correct asymptotics uniformly in $\beta$, i.e.
if $V\in L^{5/2}(\bR^3)$, we conjecture that
\be
\lim_{h\to0}{h^3 E^{\rm S}(\beta,h,V)}=-\frac{1}{15\pi^2}\int [V(x)]_+^{5/2}\rd x
\label{conj2}
\ee
holds  uniformly
in $\beta \in [0,\infty)$.
The upper bound \eqref{paulistabup} shows
that the same statement cannot hold for the Pauli case.

{\it Remark 4.}  Subleading error estimates are established first in \cite{EFS2}
and later in \cite{Iv3}.

{\it Remark 5.} The subleading error estimates in \cite{EFS2} are used in \cite{EFS3}
to give the two term energy asymptotics, i.e., up to the Scott correction, in the large nuclear charge limit
for atoms and molecules in a self-generated magnetic field.

\section{Stability bounds: Proof of Theorem~\ref{thm:stab}}
\label{sec:stab}

The lower bound in Theorem~\ref{thm:stab} establishes 
 the stability of the
system, i.e., boundedness from below of the energy. 
As the sum of
the negative eigenvalues for the Pauli operator $T^{\rm P}_h(A)-V$  goes to
minus infinity (if $[V]_+\ne 0$) as the magnetic field increases (e.g. for a constant
magnetic field \cite{AHS}), the addition of the
field energy is necessary
for stability.  Moreover, stability will also require sufficient decay of the
potential $V$ at infinity and control of the local singularities.

\medskip

 The basic tool for the proof in the Pauli case is 
the magnetic Lieb-Thirring
inequality 
from \cite{LLS}:
\begin{theorem}[Magnetic Lieb-Thirring inequality~\cite{LLS}]\label{thm:lls}
There exists a universal constant $C>0$ such that
for the Pauli operator  $T_h(A)-V$ 
with a potential $V\in L_{loc}^1(\bR^3)$ and $V_+\in L^{5/2}(\bR^3)\cap L^4(\bR^3)$ and magnetic field $B=\nabla\times A
\in L^2(\bR^3)$ we have for all $h>0$
\be
    \tr\big[ T_h(A)-V\big]_- \ge - Ch^{-3}\int \big[ V\big]_+^{5/2} 
 - C\Big( h^{-2} \int |B|^2\Big)^{3/4}
\Big( \int  \big[ V\big]_+^{4} \Big)^{1/4}.
\label{genlt}
\ee
\end{theorem}

{\it Proof of the lower bound in Theorem~\ref{thm:stab}.}
Using the magnetic Lieb-Thirring inequality \eqref{genlt}, we have
\begin{align}
   \tr \big[T_h(A)-V\big]_- & +  \beta \int |\nabla \times A|^2
\non\\ & \ge -Ch^{-3}\int [V]_+^{5/2} - Ch^{-3/2} \Big(\int [V]_+^4\Big)^{1/4}\Big(  
  \int B^2\Big)^{3/4}
+  \beta \int B^2 
\label{LLSLT}\\
&\ge -Ch^{-3}\int [V]_+^{5/2} - Ch^{-3}(\beta h)^{-3}\int [V]_+^4,
\non
\end{align}
where we set $B= \nabla\times A$ and optimized over $\int B^2$.
\qed

\subsection{Upper bound in Theorem~\ref{thm:stab}}

In order to construct a trial state that will give the upper bound in
Theorem~\ref{thm:stab} we first use the method of \cite{ES2} to show
that there exist compactly supported magnetic fields and corresponding
Pauli operators with zero-modes with arbitrarily
fast decay. The original construction of zero-modes 
in \cite{LY} will neither lead to compactly supported magnetic fields
nor arbitrarily fast decaying modes.
\begin{proposition}\label{prop:zeromodedecay} Given $m\in{\bN}$. There exist a smooth magnetic
  field of compact support $B:\bR^3\to\bR^3$ with a corresponding
  smooth vector potential
  $A:\bR^3\to\bR^3$, a smooth and non-vanishing $\psi\in C^\infty(\bR^3;\bC^2)$, and a constant $C>0$ such that 
$$
\bsigma\cdot(-i\nabla+A)\psi=0
$$
and 
$$
|\psi(x)|\leq C|x|^{-m-1}.
$$
\end{proposition}
\begin{proof}
We use the construction and notations from \cite{ES2}. Consider the map
$$
\Phi:\bR^3\ni x=(x_1,x_2,x_3)\mapsto 2\frac{x_3+i(-1+|x|^2/4)}{x_1+ix_2}\in\bC\cup\{\infty\}.
$$
This map is $\Phi=\tau_2\circ\phi\circ\tau_3^{-1}{}_{|\bR^3}$, where
$\tau_3:\bS^3\to\bR^3\cup\{\infty\}$, and
$\tau_2:\bS^2\to\bR^2\cup\{\infty\}$ are stereographic
projections (we have identified $\bC$ and
$\bR^2$) and $\phi:\bS^3\to\bS^2$ is the Hopf map (see Lemma 34 in
\cite{ES2}). Note that 
\begin{equation}\label{eq:lowerPhi}
  1+\frac14|\Phi(x)|^2=\frac{(1+|x|^2/4)^2}{x_1^2+x_2^2}\geq |x|^2/16.
\end{equation}
Consider the real 2-form on $\bC$, which by stereographic projection
pulls back to ($1/4$ times) the volume form on $\bS^2$, i.e., 
$$
\omega=\frac18(1+|z|^2/4)^{-2}idz\wedge d\overline{z}.
$$
We have $\int_\bC\omega=\pi$. For $g$ a real smooth compactly
supported function $g\in
C_0^\infty(\bC)$ we define the 2-form
$$
\beta_3=\Phi^*(g\omega)
$$
on $\bR^3$. Note that
$\beta_3$ is a closed 2-form. We can therefore define a (divergence
free) vector field $B$ such that ${B}\cdot({\bf X}\times {\bf
  Y})=\beta_3({\bf X},{\bf Y})$
for all vector fields ${\bf X}$ and ${\bf Y}$. 
By (\ref{eq:lowerPhi}), $\beta_3$ and hence $B$ has
compact support in $\bR^3$. Moreover we will assume that 
$(2\pi)^{-1}\int_\bC g\omega=m+\frac12$. 

Consider also the real 2-form $\beta_2=(g-1)\omega$ on $\bC$. It
satisfies $(2\pi)^{-1}\int_\bC \beta_2=m$. 
If we set
$$
h(z)=\pi^{-1}\int_\bC\ln|z-z'|^2\beta_2(z'),
$$
then $\beta_2=d\alpha_2$, where $\alpha_2$ is the real 1-form
$\alpha_2=2\mbox{Re}\left(\frac{i}{4}\partial_{\overline{z}}h(z)d\overline{z}\right)$.
According to the Aharonov-Casher Theorem (See Theorem 37 and Appendix
A in \cite{ES2}) the magnetic Dirac-operator on $\bC$ corresponding to the
metric with volume form $\omega$ and conformal to the standard metric
and one form (magnetic vector potential) $\alpha_2$ has a (unnormalized)
zero-mode, i.e., element in the kernel, of the form $f(z)\begin{pmatrix}1\\0
\end{pmatrix}$, where 
$$
f(z)=(1+|z|^2/4)^{1/2}\exp(-h(z)/4).
$$
In the standard metric on $\bC$ the same would be true with a zero-mode without the
prefactor $(1+|z|^2/4)^{1/2}$, which comes from the conformal factor
according to Theorem 23 in \cite{ES2}.

Note that $|f(z)|\leq C|z|^{1-m}$. In fact, there will be $m$
zero-modes, but we need only the one given above which has the fastest
decay. 

Let us turn to the construction of the zero-mode in $\bR^3$. Let ${\bf
  \Omega}:\bR^3\to\bR^3$ be the vector field corresponding to the 2-form
$\Phi^*(\omega)$, i.e., such that ${\bf \Omega}\cdot({\bf X}\times
{\bf Y})=\omega(\Phi_*({\bf X}),\Phi_*({\bf Y}))$. We may
choose a smooth unit vector field $\xi:\bR^3\to\bC^2$ such that
$\bsigma\cdot{\bf \Omega}\xi=|{\bf \Omega}|\xi$. In fact, an explicit
choice is 
$$
\xi(x)=(1+|x|^2/4)^{-1/2}\left(1+\frac{i}2\bsigma\cdot x\right)\begin{pmatrix}1\\0
\end{pmatrix}.
$$
According to
Section~8 in \cite{ES2}
we can find a smooth $A$ with $\nabla\times A=B$ and such
that $\bsigma\cdot(-i\nabla+A)\psi=0$, where
$$
\psi(x)=(1+|x|^2)^{-1}f(\Phi(x))\xi(x).
$$
Here again $(1+|x|^2)^{-1}$ is the conformal factor coming from the
stereographic projection $\tau_3$ (see Theorem 23 in \cite{ES2}).
Thus we see that $|\psi(x)|\leq C|x|^{-m-1}$.
\end{proof}
We can use this proposition to construct a low energy one-electron state
localized in a ball. 
\begin{proposition}\label{prop:compzeromode}
For all  $0<\delta<1$ there is a constant $C_\delta>0$ such that for all
$h,\beta>0$ and
all balls $B_R$ with radius $R$ we can find a smooth magnetic field $B$ supported in
$B_R$ such that for all smooth magnetic vector potentials $A$ with
$\nabla\times A=B$ in $B_R$ we can find an $L^2$-normalized $\psi\in
C_0^\infty(B_R,\bC^2)$ satisfying
$$
\int|\bsigma\cdot(-ih\nabla+A)\psi|^2+\beta\int B^2\leq C_\delta h^2\beta^{1-\delta}R^{-1-\delta}.
$$
\end{proposition}
\begin{proof}
This is a simple localization and scaling argument based on the result in the
previous proposition. First note that we may assume that $\beta R$ is
small enough. In fact, if  $\beta R>c$ for some constant $c$ then 
$h^2\beta^{1-\delta}R^{-1-\delta}\geq c^{1-\delta}h^2R^{-2}$ and an
upper bound of the form $Ch^2R^{-2}$ can be achieved by choosing
$A=B=0$. 

Without loss of generality we may assume that the ball $B_R$ is
centered at the origin. Choose an integer $m$ such that $(2m)^{-1}\leq
\delta$. Let $\widetilde\psi$ be a zeromode as constructed in
Proposition~\ref{prop:zeromodedecay},
i.e, $$\bsigma\cdot(-i\nabla+\widetilde A)\widetilde\psi=0$$ with decay
$|\widetilde\psi(x)|\leq C|x|^{-m-1}$ and corresponding magnetic field
$\widetilde B=\nabla\times \widetilde A$ of compact support. We may assume
that $\widetilde\psi$ is normalized. For $\ell>0$ we define 
$$
\psi_\ell(x)=\ell^{-3/2}\widetilde\psi(x/\ell),\quad
A_\ell(x)=\ell^{-1}h\widetilde A(x/\ell).
$$
Then $B_\ell(x)=h\ell^{-2}\widetilde B(x/\ell)$ and 
$$
\bsigma\cdot(-ih\nabla+A_\ell)\psi_\ell=0.
$$
We can assume
$\widetilde B$ to be supported in a ball of radius $1$ centered at the origin 
(otherwise we rescale as just explained). Hence
$B_\ell$ is supported in a ball of radius $\ell$.  

Choose $\chi\in C_0^\infty( \bR^3)$ with support in the unit ball
centered at the origin and such that $0\leq\chi(x)\leq 1$ for all $x$ and $\chi(x)=1$ if $|x|<1/2$. Set
$$
\psi(x)=\cN^{-1}\chi(x/R)\psi_\ell(x).
$$
where the normalization constant is $\cN=\left(\int |\chi(x/R)\psi_\ell(x)|^2\rd x\right)^{1/2}$.
Then $\psi$ is supported in $B_R$. Moreover,
$1-C(\ell/R)^{2m-1}\leq \cN^2\leq 1,$ 
and thus if $R>2\ell$
\begin{align*}
  \int|\bsigma\cdot(-ih\nabla+A_\ell)\psi|^2&=
  h^2\cN^{-2}\int(\nabla(\chi(x/R)))^2|\psi_\ell(x)|^2\rd x\\
  &\leq C h^2R^{-2}(\ell/R)^{2m-1}.
\end{align*}
Hence 
\begin{align*}
\int|\bsigma\cdot(-ih\nabla+A_\ell)\psi|^2+\beta\int
B_\ell^2&\leq C h^2R^{-2}(\ell/R)^{2m-1}+Ch^2\beta\ell^{-1}\\
&= Ch^2\beta R^{-1}(\beta R)^{-\frac1{2m}}
\end{align*}
with the optimal choice  $\ell=CR(\beta R)^{\frac1{2m}}\leq R/2$ if
$\beta R$ is small enough.

Finally, if $A$ is any smooth vector potential with $\nabla\times
A=B_\ell$ in $B_R$ then we can gauge transform, i.e., find a smooth $\phi:B_R\to\bR$ such that 
$A_\ell=A+\nabla\phi$ in $B_R$. Then 
$$
\bsigma\cdot(-ih\nabla+A)e^{ih^{-1}\phi}\psi=e^{ih^{-1}\phi}\bsigma\cdot(-ih\nabla+A_\ell)\psi
$$
and thus the above bound holds with $\psi$ replaced by $e^{-ih^{-1}\phi}\psi$.
\end{proof}

\begin{proof}[Proof of the upper bound in Theorem~\ref{thm:stab}.]
By choosing $A=0$ we can always achieve a Weyl upper bound 
$$
h^3E(\beta, h, V) \leq-C\int[V]_+^{5/2}+\cE_V(h).
$$
Let us now show that if $(\beta h)^{-3+2\varepsilon}\int
[V]_+^{4-\varepsilon} > \int[V]_+^{5/2}$ then we can achieve the bound 
\begin{equation}\label{eq:smallbhbound}
  h^3E(\beta, h, V) \leq-C_\varepsilon(\beta h)^{-3+2\varepsilon}\int [V]_+^{4-\varepsilon}
\end{equation}
for $h$ small enough depending on $V$ and $\varepsilon$.  
Divide space into cubes of side length $\sqrt{h}$. For each cube consider the
minimal value $V_{\min}$ of $V$. 
Then in the cube $V\leq
V_{\min}+\sqrt{3h}\|\nabla V\|_\infty$. If $V_{\min}\leq \sqrt{h}\|\nabla
V\|_\infty$ we do nothing in the cube. If we denote the union of all
these cubes $WQ$ (for weak cubes) we find that 
\begin{equation}\label{eq:choiceh}
  \int_{WQ}[V]_+^{4-\varepsilon}\leq\
  C_Vh^{2-\varepsilon/2}
\end{equation}
for $C_V>0$ a constant depending only on $V$ (in particular on the
support of $V$). 

In each cube where $V_{\min}\geq \sqrt{h}\|\nabla V\|_\infty$ we have 
$V\leq (1+\sqrt{3})V_{\min}$. Each of these cubes
we fill with the maximal number of disjoint balls of radius 
$$
R=\kappa_2h(\beta h)^{1-2\varepsilon/3}V_{\min}^{-1+\varepsilon/3},
$$
where $\kappa_2>0$ is a constant which we will choose below depending
only on $\varepsilon$. 
Note that 
\begin{equation}\label{eq:choiceh2}
R\leq \kappa_2 ^{-1+\varepsilon/3}h^{1/2+\varepsilon/6}(\beta
h)^{1-2\varepsilon/3}\|\nabla V\|^{-1+\varepsilon/3}_\infty\leq
\sqrt{h}/2
\end{equation}
if $h$ is small enough depending
only on $V$ and $\varepsilon$ (recall that $\beta h$ is bounded from above in terms
of $V$). In particular, we can then fit at least one ball in the cube. 

In each of these balls we choose a magnetic field according to
Proposition~\ref{prop:compzeromode} with $\delta$ chosen such that 
$3\delta(1+\delta)^{-1}=\varepsilon$.
Let $B:\bR^3\to\bR^3$ be the
sum of all these disjointly supported magnetic fields and let 
$A$ be a corresponding vector potential. Let $B^{(i)}$and $R_i$ for $i=1,2,\ldots$ denote the
balls (there are only finitely many) and their radii. We can according to
Proposition~\ref{prop:compzeromode} for each $i$
find a normalized $\psi_i$ supported in $B^{(i)}$ such that 
\begin{align*}
  \int|\bsigma\cdot(-ih\nabla+A)\psi_{i}|^2-\int |\psi_i|^2V+\beta\int_{B^{(i)}}
  B^2&\leq C_\delta h^2\beta^{1-\delta} R_i^{-1-\delta}-V_{\min,i}\\&\leq
  (C_\delta\kappa_2^{-1-\delta}-1)V_{\min,i}\\&\leq -V_{\min,i}/2
\end{align*}
if $\kappa_2$ is large enough depending on $\delta$ , i.e., on
$\varepsilon$. Here $V_{\min,i}$ is the minimum of $V$ in the cube
containing the ball $B^{(i)}$. 

Since all the constructed balls are disjoint we have that 
$P=\sum_i|\psi_i\rangle\langle\psi_i|$ is an orthogonal projection and
hence 
\begin{align*}
E(\beta,h,V)&\leq
\Tr\left[\left((\bsigma\cdot(-ih\nabla+A))^2-V\right)P\right]+\beta\int
B^2\\
&\leq\sum_i-V_{\min,i}/2=\sum_i-\frac12\kappa_2^{-3}h^{-3}(\beta
h)^{-3+2\varepsilon}V_{\min,i}^{4-\varepsilon}R_i^3\\
&\leq-2C_\varepsilon h^{-3}(\beta
h)^{-3+2\varepsilon}\int_{\bR^3\setminus WQ}V^{4-\varepsilon},
\end{align*}
for some constant $C_\varepsilon>0$ depending only on $\varepsilon$. 
We have here used that in each of the cubes the balls take up a
certain fraction bounded below of the volume and that $V\leq
(1+\sqrt{3})V_{\min}$.  
Let us emphasize that $\kappa_2$ was chosen depending only on
$\varepsilon$.  This allows us to ensure that (\ref{eq:choiceh2}) is
satisfied if $h$ is small enough depending on $V$ and $\varepsilon$.  Finally, from 
(\ref{eq:choiceh}) we can choose $h$ so small depending on $V$ and $\varepsilon$ that 
$$
\int_{WQ}[V]_+^{4-\varepsilon}\leq \frac12 \int_{\bR^3}[V]_+^{4-\varepsilon}.
$$
This proves the claim (\ref{eq:smallbhbound}). 
\end{proof}

\section{Semiclassics for weak fields: Proof of
Theorem~\ref{thm:weak} }\label{sec:lower}

{\it Proof of Theorem~\ref{thm:weak}.} We discuss only the Pauli case,
the Schr\"odinger case is similar but  much easier and is left to the
reader. In the remaining part of the proof we will omit the
superscript P.

For the upper bound on $ h^3 E(\beta, h, V)$ we choose $A\equiv 0$
in the definition \eqref{eq:energy} and then we have
$E(\beta, h, V)\le E(\infty, h, V)$ and the second equality in
\eqref{eq:asymppauli} is just the usual non-magnetic semiclassical
asymptotics. 

For the lower bound, we first remark $V\ge 0$ can be assumed and that it is sufficient to prove
the result for $V\in C_0^\infty(\bR^3)$ by a standard approximation argument.
The error between $V\in L^{5/2}\cap L^4$ and its $C_0^\infty$-approximation
$\wt V$ can be made arbitrarily small in the $\|\cdot \|_{5/2} + \|\cdot \|_4$ norm.
Thus the replacement of $V$ with $\wt V$ can be controlled by
using  the magnetic Lieb-Thirring inequality \eqref{genlt}
and borrowing a small part of the kinetic energy and 
the magnetic energy. For more details, see Section 5.4 of \cite{ES3}
(with the only modification that instead of (5.58) of \cite{ES3}
use \eqref{LLSLT}).

Secondly, as discussed in Appendix~\ref{vectorfields}
we may in \eqref{eq:energy}  replace $\int_{\bR^3}|\nabla\times A|^2$ 
by  $\int_{\bR^3}|\nabla\otimes A|^2$, where the last integrand contains all derivatives.

Thirdly, we may replace $E(\beta, h, V)$ with
a localized version of the total energy. 
Let $0\le \phi^*(x)\le 1$ be a smooth function with $\mbox{supp}\, \phi^*\subset
B(1)$ and $\phi^*\equiv 1$ on $B(1/2)$,
 where $B(r)$ denotes the ball of radius $r$ centered at the origin.
Denote by $\phi_r(x)=\phi^*(x/r)$. 
Using a partition of unity, $\phi_r^2 + \eta_r^2\equiv 1$, 
and the IMS localization, we have 
\be
   \Tr [ T_h(A)-V]_- \ge \tr \big[ \phi_r(T_h(A)-V - h^2 I_r)\phi_r\big]_-
 +  \tr \big[ \eta_r\big(T_h(A)-V - h^2I_r\big)\eta_r\big]_-,
\label{IMSS}
\ee
where $I_r: =(\nabla \phi_r)^2 +(\nabla\eta_r)^2$
is supported in the shell $\{ r/2\le |x|\le r\}$.
The second term is bounded by the magnetic Lieb-Thirring inequality
similarly to \eqref{LLSLT}. More precisely, for any $\e>0$
there is a sufficiently large $r=r_\e\ge 1$ such that
\begin{align}
   \tr \big[ \eta_r\big(T_h(A)-V - h^2I_r\big)\eta_r\big]_-
  \ge & \;\tr  \big(T_h(A)-\wh V)_-
\non\\
\ge & - Ch^{-3}\int_{|x|\ge r} \wh V^{5/2}
  - Ch^{-3}(\beta h)^{-3} \Big(\int_{|x|\ge r} 
\wh V^{4}\Big) - \frac{\beta}{2}  \int |\nabla\times A|^2
\non \\
  = & \; - \e h^{-3}  - \frac{\beta}{2}  \int |\nabla\times A|^2
\end{align}
holds if $h \le h_\e$, 
where $\wh V(x) : =  V(x) + h^2I_r$.
Here we used the integrability conditions on $V$
and  that $\beta h$ is bounded from below.

It is therefore sufficient to give a lower bound on the first
term in \eqref{IMSS}, more precisely, 
we have
\be
E(\beta, h, V) \ge -\e h^{-3} + 
\inf_{A} \Big[\tr \big[ \phi(T_h(A)-W)\phi\big]_- + 
\frac{\beta}{2} \int_{\bR^3} |\nabla \otimes A|^2\Big],
\label{locest}
\ee
where we set $W:= V + h^2I_r$ and $\phi=\phi_r$ for brevity.

To estimate the right hand side of \eqref{locest},
we will follow the argument of Section~5 of \cite{ES3}.
We choose a length $L$ with
$ h\le L \le h^{1/2}$.
Let $\Omega_L: = B(r+L)$ be the $L$-neighborhood of $\Omega:=B(r)$.
Let $Q_k:=\{ y\in \bR^3\; : \; \| y-k\|_\infty < L/2\}$
with $k\in (L\bZ)^3\cap \Omega_L$ denote  a non-overlapping
covering  of $B(r)$
 with boxes of size $L$.
In this section
the index $k$ will always run over the set $(L\bZ)^3\cap \Omega_L$.
Let $\xi_k$ be a partition of unity, $\sum_k \xi^2_k\equiv 1$, subordinated to 
the collection of boxes $Q_k$, such that
$$
    \mbox{supp}\;\xi_k \subset (2Q)_k, \qquad
    |\nabla \xi_k|\leq CL^{-1},
$$
where $(2Q)_k$ denotes the cube of side-length $2L$
with center $k$. Let $\wt\xi_k$  be a cutoff
function such that $\wt\xi_k \equiv 1$ on $(2Q)_k$ (i.e. on the
support of $\xi_k$),
$\mbox{supp}\, \wt\xi_k\subset \wt Q_k:=(3Q)_k$
and $|\nabla \wt\xi_k|\leq CL^{-1}$.

Let $\langle A\rangle_k = |\wt Q_k|^{-1}\int_{\wt Q_k} A$,
$A_k: = (A- \langle A\rangle_k)\wt\xi_k$ and
$B_k : = \nabla\times A_k$,
then by Poincar\'e inequality we have
\begin{align}
\int_{\bR^3} B_k^2 \le  \int_{\wt Q_k} |\nabla \otimes A_k|^2 
 \le & \; C  \int_{\wt Q_k} |\nabla \otimes A|^2 + C L^{-2} \int_{\wt Q_k} |A- \langle A\rangle_k |^2 \non\\  
\leq & \; C\int_{\wt Q_k} |\nabla \otimes A|^2.
\label{ba}
\end{align}
{F}rom the IMS localization with a phase function $\psi_k$ satisfying 
$h\nabla\psi_k = \langle A_k\rangle$ we have
\be
\begin{split}
\Tr\big[  \phi ( T_h(A) -W)\phi \big]_- + \frac{\beta}{2}\int_{\bR^3} B^2 
= & \inf_\gamma  \Tr\Big( \gamma \phi[T_h(A)-W]\phi\Big)
 +\beta \int_{\bR^3}|\nabla \otimes A|^2 \cr
 \ge &  \inf_\gamma\sum_{k\in(L\bZ)^3\cap \Omega_L} \cE_k(\gamma)
\end{split}\label{qk}
\ee
with
$$
  \cE_k(\gamma):= \Tr \Big[ \gamma \xi_k e^{-i\psi_k}
\phi[T_h(A-\langle A\rangle_k)-W - Ch^2L^{-2}]\phi e^{i\psi_k}\xi_k 
\Big]
+c_0\beta\int_{\wt Q_k}|\nabla \otimes A|^2
$$
with some universal constant $c_0$ and after reallocating the localization error.
In \eqref{qk} the infimum is taken  over all density matrices $0\leq \gamma\leq 1$.
We also  reallocated the second integral to account for the finite
overlap of the cubes $\wt Q_k$. We introduce the notation
$$
  \cF_k:=c_0\beta \int_{\wt Q_k}|\nabla \otimes A|^2.
$$

Let $[H]_Q$ denote the operator $H$ with Dirichlet boundary conditions on the box $Q$.
For each fixed   box $\wt Q_k$ we apply  the magnetic Lieb-Thirring inequality
\cite{LLS} together with \eqref{ba} 
to obtain that for any density matrix
$\gamma$
\be
\begin{split}
\cE_k(\gamma) &\ge  \Tr\Big[  [T_h(A_k)-W- Ch^2L^{-2}]_{\wt Q_k} \Big]_- 
+\cF_k \non \cr
 & \ge - Ch^{-3}\int_{\wt Q_k}[ W+ Ch^2L^{-2}]^{5/2}
  - C \Big( \int_{\wt Q_k} [ W+ Ch^2L^{-2}]^4\Big)^{1/4}
    \Big( h^{-2}\int_{\wt Q_k} B_k^2\Big)^{3/4}+\cF_k \non \cr
& \ge - Ch^{-3}L^3 - Ch^{-6}L^3 \beta^{-3}
 -  \frac{c_0}{2}\,\beta\int_{\wt Q_k} |\nabla\otimes A|^2+\cF_k\non\cr
&\ge - Ch^{-3}L^3+\frac{1}{2}\cF_k 
\end{split}
\ee
using $h\le L$ and $\beta h\to \infty$. The constants
 $C$ depend on $\| W\|_\infty$.

Let $S\subset (L\bZ)^3\cap  \Omega_L$ 
denote the set of those $k$ indices 
such that 
\be
\begin{split}
 \cF_k
& \leq  C_1h^{-3}L^3.
\end{split}\label{aprio}
\ee
holds with some large constant $C_1$. In 
 particular, by choosing $C_1$ sufficiently large, we have 
\be
  \cE_k(\gamma)\ge 0, \quad \mbox{for all $k\not\in S$ and for
any $\gamma$}.
\label{cebig}
\ee

We use the Schwarz inequality
in the form
$$
T_h(A-\langle A\rangle_k) \ge -(1-\e_k)h^2\Delta - C\e_k^{-1}(A-\langle
A\rangle_k)^2,
$$
 with some $0<\e_k<\frac{1}{3}$, then
\be
\begin{split}
 \cE_k(\gamma)\ge &
\Tr\Big[\phi\xi_k 
[-(1-2\e_k)h^2\Delta-W-Ch^2L^{-2}]\xi_k\phi\Big]_-\cr & +
\Tr \Big[{\bf 1}_{\wt Q_k} [-\e_k h^2\Delta-C\e_k^{-1}(A-\langle
A\rangle_k)^2] {\bf 1}_{\wt Q_k}\Big]_- +\cF_k.\cr
\end{split}
\ee
We will show at the end of the section that
\be
\begin{split}
  \Tr\Big[ \phi\xi_k [-(1&-2\e_k)h^2\Delta -W-Ch^2L^{-2}]\xi_k
\phi\Big]_-\cr
  &\ge \Tr \Big[ \phi\xi_k(-h^2\Delta -W)\xi_k\phi\Big]_- 
- Ch^{-3}
  (\e_k + h^2L^{-2}) |\wt Q_k|.
\label{stand}
\end{split}
\ee
Using \eqref{cebig} and \eqref{stand},
\be
\begin{split}\label{schw}
\inf_\gamma\sum_k \cE_k(\gamma)\ge & \inf_\gamma\sum_{k\in S} 
\cE_k(\gamma) \cr
\ge&
 \sum_k \Tr\Big[\phi\xi_k [-h^2\Delta-W]\xi_k\phi\Big]_-
 + \sum_{k\in S} \cD_k\cr
\ge&  
  \sum_k \inf_{\gamma_k}\Tr\Big[ \xi_k \gamma_k\xi_k\phi
[-h^2\Delta-W]\phi\Big]
+ \sum_{k\in S} \cD_k\cr
\ge & \; \Tr \big[ \phi(-h^2\Delta-W)\phi
\big]_- +\sum_{k\in S} \cD_k
\end{split}
\ee
with
\be
 \cD_k:=
\Tr \Big[  [-\e_k h^2\Delta-
C\e_k^{-1}(A-\langle A\rangle_k)^2]_{\wt Q_k}\Big]_- 
- Ch^{-3}|\wt Q_k| (\e_k + h^2L^{-2})
+\cF_k .
\label{ce}
\ee
In the last step in \eqref{schw} 
we used that for any collection of density matrices
$\gamma_k$, the density matrix $\sum_k \xi_k \gamma_k
\xi_k$ is
 admissible  in the variational
principle 
\be
\Tr \big[ \phi(  -h^2\Delta  -W)\phi
\big]_- =
\inf\Big\{\Tr\gamma\big[-h^2\Delta-W\big]\; :
 \; 0\leq \gamma\leq 1, \Big\}.
\label{varprin}
\ee

We estimate $\cD_k$ for $k\in S$ as follows
\be\label{DD}
\begin{split}
\cD_k
\ge & -C\e_k^{-4} h^{-3}
 \int_{\wt Q_k} (A-\langle A\rangle_k)^5 
- Ch^{-3}|\wt Q_k|  (\e_k + h^2L^{-2})
 +\cF_k\cr
\ge &\;  \cF_k -C\e_k^{-4}h^{-3}\beta^{-5/2}L^{1/2}\cF_k^{5/2}
 - Ch^{-3} |\wt Q_k| (\e_k + h^2L^{-2}).
\end{split}
\ee
In the first step we used Lieb-Thirring inequality, in the second
step  H\"older and Sobolev inequalities
in the form
$$
   \int_{\wt Q_k} (A-\langle A\rangle_k)^5 \leq
  CL^{1/2} 
\Big(\int_{\wt Q_k}|\nabla\otimes A|^2\Big)^{5/2}.
$$
We choose
$$
  \e_k = \beta^{-1/2}L^{-1/2}\cF^{1/2}_k
$$
and using the a priori bound \eqref{aprio}, we see that
$$
  \e_k \leq Ch^{-1}L (\beta h)^{-1/2}.
$$
Thus, choosing
\be
   L = h(\beta h)^{1/10},
\label{Lcond}
\ee
we get $\e_k \le C(\beta h)^{-2/5}\leq 1/3$ as $\beta h\to \infty$.
With this choice of $\e_k$ and $L$  we have from \eqref{DD}
\be
  \cD_k 
\ge \cF_k - C (h\beta)^{-1/4}\cF_k^{1/2} - C h^{-3}L^3 (\beta h)^{-1/5}
\ge - C h^{-3} L^3 (\beta h)^{-1/5}.
\label{d1}
\ee
Summing up \eqref{d1} for all $k$ and using that 
$$
  \sum_{k\in (L\bZ)^3\cap \Omega_L} L^3 \leq Cr^3,
$$
we obtain
from  \eqref{schw} and \eqref{d1} 
\begin{align}
\inf_\gamma\sum_k \cE_k(\gamma)\ge &
\;\Tr \Big[\phi(-h^2\Delta-W)\phi\Big]_- - Ch^{-3}(\beta h)^{-1/5}r^3 \non \\
 \ge & -\frac{2}{15\pi^2} (1+o(1)) \int W^{5/2}  - Ch^{-3}(\beta h)^{-1/5}r^3 \label{schw3}
\end{align}
using the standard semiclassical asymptotics for $\Tr \big[\phi(-h^2\Delta-W)\phi\big]_- \ge
\Tr (-h^2\Delta-W)_-$. 
Together with \eqref{qk}
 this proves the required lower bound 
for the second term in \eqref{locest}. The difference between
$\int W^{5/2}$ and $\int V^{5/2}$ is negligible as $h\to 0$.
 Letting first $h\to 0$
together with $\beta h\to \infty$ and then letting $\e\to 0$
we obtain the lower bound in \eqref{eq:asymppauli}.

\bigskip

Finally, we prove \eqref{stand}.
 Let $\gamma$ be a trial density matrix for
the left hand side of \eqref{stand}. We can assume
that
$$
  0\ge \Tr\Big[\gamma \phi
\xi_k [-(1-2\e_k)h^2\Delta -W-Ch^2L^{-2}]\xi_k \phi\Big].
$$
Recalling that $\e_k\le \frac{1}{3}$, we have 
\be
\begin{split}
 0 \ge & \Tr\Big[\gamma \phi
\xi_k [-\frac{1}{6}h^2\Delta +1]\xi_k \phi\Big] 
+ \Tr\Big[\gamma\phi
\xi_k [-\frac{1}{6}h^2\Delta -W-Ch^2L^{-2}-1]\xi_k \phi\Big]\cr
\ge  &\Tr\Big[\gamma \phi
\xi_k [-\frac{1}{6}h^2\Delta +1]\xi_k \phi\Big] 
 - Ch^{-3}\int_{\wt Q_k} [ W+ 1 + Ch^2L^{-2}]^{5/2} ,
\end{split}
\ee
where we used Lieb-Thirring inequality. Thus, using  $h\le L$, we have
$$
  \Tr\Big[\gamma \phi
\xi_k [-\frac{1}{6}h^2\Delta +1]\xi_k \phi\Big] 
 \leq Ch^{-3}|\wt Q_k|
$$
with a constant depending on $W$.
Therefore
\be
\begin{split}
  \Tr\Big[ \gamma
\phi\xi_k [-(1&-2\e_k)h^2\Delta -W-Ch^2L^{-2}]\xi_k
\phi\Big]\cr
  &\ge \Tr \Big[ \gamma
\phi\xi_k(-h^2\Delta -W)\xi_k\phi\Big]
 - Ch^{-3}
  (\e_k + h^2L^{-2}) |\wt Q_k|.
\label{stand3}
\end{split}
\ee
Now  \eqref{stand} follows
by the variational principle.
$\Box$

\appendix
\section{Equivalent forms of the energy}
\label{vectorfields}

We will consider the equivalence of the total energy where we have different restrictions on the vector potentials. We allow the energy to possibly have an extra localization. So we end up considering
\begin{align}
E_{\times}(A) = \tr\big[ \psi(T_h(A) - V) \psi ]_{-}
 + \beta \int_{\bR^3} |\nabla \times A|^2,
\end{align}
where $\psi \in C^{\infty}(\bR^3)$ satisfies $0\leq \psi \leq 1$.
Similarly, we define
\begin{align}
E_{\otimes}(A) = \tr\big[ \psi(T_h(A) - V) \psi ]_{-}+ \beta \int_{\bR^3} |\nabla \otimes A|^2,
\end{align}
where $|\nabla \otimes A|^2=\sum_{i,j=1}^3|\partial_iA_j|^2$.
Some natural domains of definition are given below:
\begin{align*}
{\mathcal D}_1 &= \Big\{ A \in L^6(\bR^3,\bR^3)\,:\, 
\int |\nabla \times A|^2 < \infty \Big\}, \\
{\mathcal D}_2 &= C_0^{\infty}(\bR^3,\bR^3),\\
{\mathcal D}_3 &= H^1(\bR^3,\bR^3),\\
{\mathcal D}_4&=\{ A \in {\mathcal D}_1\,:\,\nabla \cdot A = 0\}.
\end{align*}
In the case of $E_{\otimes}$, the expression $\int |\nabla \times A|^2$ in ${\mathcal D}_1$ 
should be replaced by $\int |\nabla \otimes A|^2$.

We will only assume that $V \in L^1_{\rm loc}(\bR^3)$. The trace in the above expressions 
should then be interpreted as 
\begin{align}\label{eq:TraceFriedrichs}
\inf \sum_{j=1}^N \langle \phi_j | \psi(T_h(A) - V) \psi \phi_j \rangle,
\end{align}
where $\{\phi_j\}_{j=1}^N$ runs over all orthonormal subsets of $C_0^{\infty}(\bR^3)$.
 If this infimum is different from $-\infty$, it implies in particular that the 
quadratic form of $\psi(T_h(A) - V) \psi$ defined on $C_0^{\infty}$ is semibounded 
from below. In that case \eqref{eq:TraceFriedrichs} will be equal to the trace of 
the negative part of the Friedrichs extension of this quadratic form, thereby justifying the notation.

\begin{proposition}
We have for all $i,j \in \{ 1, 2, 3, 4\}$,
\begin{align}
\inf_{A \in {\mathcal D}_i} E_{\times}(A) = 
\inf_{A \in {\mathcal D}_j} E_{\otimes}(A). \end{align}
\end{proposition}

Notice though that we do not prove that one can both impose compact support and 
zero divergence at the same time.

\begin{proof}
Consider first the $E_{\times}$. Clearly, ${\mathcal D}_2 \subset {\mathcal D}_3 
\subset {\mathcal D}_1$ (using the Sobolev inequality to get the last inclusion)
 which implies corresponding inequalities for the energies. We will now prove that
 $\inf_{A \in {\mathcal D}_1} E_{\times}(A) \geq \inf_{A \in {\mathcal D}_2} E_{\times}(A)$. 
But for any $A \in {\mathcal D}_1$ and any finite collection $\{ \phi_j \} \subset C_0^{\infty}(\bR^3)$ 
we can get arbitrarily close to 
\begin{align*}
\sum_{j=1}^N \langle \phi_j | \psi(T_h(A) - V) \psi \phi_j \rangle + \beta \int |\nabla \times A|^2
\end{align*}
by simultaneously approximating $A$ in $L^6$-norm and $\nabla \times A$ in $L^2$, by a
 $C_0^{\infty}$ vector field. Therefore
\begin{align*}
\inf_{A \in {\mathcal D}_1} E_{\times}(A) = \inf_{A \in {\mathcal D}_2} E_{\times}(A) 
=\inf_{A \in {\mathcal D}_3} E_{\times}(A).
\end{align*}

Clearly $\inf_{A \in {\mathcal D}_1} E_{\times}(A) \leq \inf_{A \in {\mathcal D}_4} E_{\times}(A)$.
We will prove that $\inf_{A \in {\mathcal D}_2} E_{\times}(A) \geq \inf_{A \in {\mathcal D}_4} E_{\times}(A)$,
 thereby establishing the equality for all four energies $E_{\times}$.
 Let $A \in C_0^{\infty}(\bR^3,\bR^3)$ and $\{ \phi_j \}_{j=1}^N \subset C_0^{\infty}(\bR^3)$.
 Then $B=\nabla \times A \in L^2(\bR^3)$ and therefore there exists $A' \in {\mathcal D}_4$ 
with $\nabla \times A' = B$ (see \cite{FLL}). It follows that there exists $\eta$ 
with $A-A' = \nabla \eta$ and therefore (since $\Delta \eta = \nabla \cdot A$), 
$\eta \in C^{\infty}(\bR^3)$. But then
\begin{align*}
\sum_{j=1}^N \langle \phi_j | \psi(T_h(A) - V) \psi \phi_j \rangle = 
\sum_{j=1}^N \langle e^{i\eta} \phi_j | \psi(T_h(A') - V) \psi (e^{i\eta}\phi_j) \rangle,
\end{align*}
which establishes the desired inequality.
Since $\int |\nabla \otimes A|^2 = \int |\nabla \times A|^2 + \int |\nabla \cdot A|^2$, 
the same arguments give the identities for the $E_{\otimes}$ versions of the energies.

Finally we prove that $\inf_{A \in {\mathcal D}_4} E_{\times}(A) = \inf_{A \in {\mathcal D}_4}
 E_{\otimes}(A)$. But this is obvious since the field energies are identical when $\nabla \cdot A=0$.
\end{proof}

\section{Self-generated magnetic fields lower the energy}\label{sec:paramagnetism}

In this appendix we will show that self-generated magnetic fields may
indeed decrease the energy, i.e., inequality (\ref{eq:paramagnetism}).

{\it Proof of (\ref{eq:paramagnetism}).} For the Pauli operator we
already remarked this fact as a consequence of
Theorem~\ref{thm:stab}. Alternatively, it also follows from the instability in
\eqref{instab} in Appendix~\ref{app:stab} below, since the non-magnetic Hydrogen atom is stable. For the
Schr\"odinger operator this statement was essentially proved in
\cite{ELV} (see also \cite{FLW}) by considering the perturbative
regime as a small magnetic field is turned on.  A simple first order
perturbation argument shows that the lowest eigenvalue increases
quadratically in $B$. In a spherical geometry the higher non-magnetic
eigenvalues are degenerate and some of them carry non-trivial
current. These eigenvalues will split linearly when a small magnetic
field is turned on.  To see this explicitly we can consider a
spherically symmetric harmonic oscillator in a constant magnetic
field, i.e., $V(x)=|x|^2$ and $A(x,y,z)=(By/2,-Bx/2,0)$ with $B>0$
constant.  The eigenvalues of the operator $(-i\nabla-A)^2+|x|^2$ are
(see \cite{Fo})
$$
e(n_1,n_2,n_3)=(n_1+n_2+1)\sqrt{1+B^2} +(n_3+1/2)+(n_1-n_2)B
$$
with $n_1, n_2, n_3\in \bN$.
Thus as an explicit example 
$$
\tr ((-i\nabla-A)^2+|x|^2-5/2)_-=3\sqrt{1+B^2}-4-B
$$
which of course explicitly decreases as a small $B$ is increased from zero.
It is now clear that we can find 
$\widetilde A\in C_0^\infty(\bR^3;\bR^3)$ which approximates $A$ such that 
$$\tr ((-i\nabla-\widetilde A)^2+|x|^2-5/2)_-<
\tr (-\Delta+|x|^2-5/2)_- =-1
$$ and hence for $\beta>0$ sufficiently small
$$
\tr ((-i\nabla-\widetilde A)^2+|x|^2-5/2)_-+\beta\int
|\nabla\times\widetilde A|^2<\tr [-\Delta+|x|^2-5/2]_-.
$$
\qed

\section{Stability conditions}\label{app:stab}

Using the argument in  \cite{FLL} it is easy to show the following
stability result on the one-electron energy.
\begin{proposition} Let $V\in L^1_{loc}(\bR^3)$ with $V_+\in
  L^3(\bR^3)\cap L^{3/2}(\bR^3)$. Then for all $\psi\in H^1(\bR^3)$
  with $\|\psi\|_{L^2}=1$ we have 
  $$
  \langle \psi,(T_h^{\rm P}(A)-V)\psi\rangle+\beta\int B^2\geq0
  $$
  if $\beta^{-1}h^{-2}\|V_+\|_3$ and $h^{-2}\|V_+\|_{3/2}$ are (universally) small enough. 
\end{proposition}
\begin{proof}
Let $C_S>0$ be the Sobolev constant,  i.e., 
$\int|\nabla\psi|^2\geq C_S\|\psi\|_6^2$.
Since $$T^{\rm P}_h(A)=(-ih\nabla+A)^2+h\bsigma\cdot B$$ 
we estimate for all $0<\varepsilon\leq1$
\begin{align*}
  \langle \psi,T_h^{\rm P}(A)\psi\rangle+\beta\int B^2&\geq 
  C_Sh^2\varepsilon\|\psi\|_6^2-h\varepsilon\int|B||\psi|^2+\beta\int B^2\\&\geq 
  C_Sh^2\varepsilon\|\psi\|_6^2-(4\beta)^{-1}\varepsilon^2h^2\int|\psi|^4\\
  &\geq C_Sh^2\varepsilon\|\psi\|_6^2-(4\beta)^{-1}\varepsilon^2h^2\|\psi\|_6^{3}\|\psi\|_2.
\end{align*}
We will also use that for $p\geq 3/2$ we have the H\"older inequality
$$
\int V|\psi|^2\leq \|V_+\|_p\|\psi\|_6^{3/p}\|\psi\|_2^{2-3/p}.
$$
We consider two cases.

{\bf Case 1:} $\|\psi\|_6\leq 2C_S\beta$. We set $\varepsilon=1$ and
$p=3/2$ above
and find since $\|\psi\|_2=1$
$$
\langle \psi,(T_h^{\rm P}(A)-V)\psi\rangle+\beta\int B^2\geq \frac12
C_sh^2\|\psi\|_6^2-\|V_+\|_{3/2}\|\psi\|_6^2
$$
from which it follows that the energy is non-negative if
$\|V_+\|_{3/2}\leq C_Sh^2/2$. 
\medskip

{\bf Case 2:} $\|\psi\|_6\geq 2C_S\beta$. Let
$\varepsilon=2C_S\beta\|\psi\|_6^{-1}\leq 1$ and $p=3$. 
Then 
$$
\langle \psi,(T_h^{\rm P}(A)-V)\psi\rangle+\beta\int B^2\geq
C_S^2h^2\beta\|\psi\|_6-\|V_+\|_3\|\psi\|_6.
$$
Hence the energy is non-negative if $\|V_+\|_3\leq C_S^2h^2\beta$.
\end{proof}

It follows immediately from this proposition that the one-electron
energy $E^{\rm P}_0(\beta,h,V)$ is finite if $V_+\in L^3(\bR^3)$. In
fact, all we have to argue is that $\beta h^{-2}\|[V-e]_+\|_3$ and
$h^{-2}\|[V-e]_+\|_{3/2}$ can be made small enough by choosing $e>0$
large enough. In this way $-e$ can be made a lower bound on $E_0^{\rm
  P}$. Since $V_+\in L^3$ we can of course make $\|[V-e]_+\|_3$
arbitrarily small. Using that $[V-e]_+^{3/2}\leq
(2e^{-1})^{3/2}[V-e/2]_+^3$ we can do the same with
$\|[V-e]_+\|_{3/2}$.

This stability criterion is essentially sharp. In fact, applying 
the method of proof as in the proposition above and the
construction of zero-modes in \cite{LY} it was proved in \cite{FLL}
that for the Coulomb potential $V(x)=c|x|^{-1}$ there is
a critical value $\gamma_{\rm cr}$ such that the one-electron energy
satisfies \be E_0^{\rm P}(\beta, h, V)>-\infty \qquad \mbox{if} \;\;
\gamma_{\rm cr}\beta h^2> c
\label{stabh0}
\ee
and
\be
    E^{\rm P}_0(\beta, h, V)= -\infty \qquad \mbox{if} \;\; \gamma_{\rm
      cr}\beta h^2< c.
\label{instab0}
\ee
Since $E^{\rm P}\leq E^{\rm P}_0$ it is clear that (\ref{instab0})
implies that even for the cutoff Coulomb potential $V=[c|x|^{-1}-1]_+$
we have 
\be
    E^{\rm P}(\beta, h, V)= -\infty \qquad \mbox{if} \;\; \gamma_{\rm
      cr}\beta h^2< c.
\label{instab}
\ee
However there is also a value $\gamma'_{\rm cr}>0$ such that
 \be
    E^{\rm P}(\beta, h, V)>-\infty \qquad \mbox{if} \;\; \gamma'_{\rm
      cr}\beta h^2> c
\label{stabh}
\ee 
This stability statement follows, e.g., by localizing in an appropriate ball 
and then follow the proof of Lemma 2.1 \cite{ES3} for the inner regime
(with the choice of $Z=h^{-2}$, $\delta= Z^{1/3}$, $D=RZ^{1/3}$
and $16\pi\al^2= \beta^{-1}$). In the outer regime the operator
has a compactly supported bounded potential (that includes the localization
error) so its energy is controlled by the magnetic 
Lieb-Thirring inequality as 
in \eqref{LLSLT}.


\begin{thebibliography}{hhhhh}


\bibitem[AHS]{AHS}  J. Avron, I. Herbst and B. Simon: {\em Schr\"odinger
operators with magnetic fields. I. General interactions. \/} Duke Math. 
J. {\bf 45} (1978), 847--883. 



\bibitem[ELV]{ELV}
L. Erd{\H o}s, M. Loss and V. Vougalter,
{\it Diamagnetic behavior of sums of Dirichlet eigenvalues.}
 Ann. Inst. Fourier (Grenoble), Vol {\bf 50}, no. 3. 891-907 (2000).


\bibitem[ES1]{ES1} L. Erd{\H o}s and J. P. Solovej: {\it Semiclassical
eigenvalue estimates for the Pauli operator with strong
non-homogeneous magnetic fields. II. Leading order asymptotic estimates.}
Commun. Math. Phys. {\bf 188}, 599--656 (1997)






\bibitem[ES2]{ES2} 
L. Erd{\H o}s, J. P. Solovej, {\it 
The kernel of Dirac operators on $S^3$ and ${\bf R}^3$.}
Rev. Math. Phys. {\bf 13} No. 10, 1247-1280 (2001)



\bibitem[ES3]{ES3}  L. Erd{\H o}s, J. P. Solovej, 
{\em Ground state energy of large atoms in a self-generated
magnetic field.} Commun. Math. Phys. {\bf 294}, No. 1, 229-249 (2009)




\bibitem[EFS2]{EFS2}  L. Erd{\H o}s, S. Fournais,
J.P. Solovej: {\em Second order semiclassics with self-generated
magnetic fields.}  Preprint: arxiv.org/1105.0512 


\bibitem[EFS3]{EFS3}  L. Erd{\H o}s, S. Fournais, J.P. Solovej: 
{\em Scott correction for large molecules with a self-generated magnetic
field.}   Preprint: arxiv.org/1105.0521


\bibitem[FS]{FS} C.~Fefferman and L.A.~Seco: {\em On the energy of a
    large atom}, Bull.~AMS {\bf 23}, 2, 525--530 (1990).

\bibitem[Fo]{Fo} V. Fock, 
{\em Bemerkung zur Quantelung des harmonischen 
Oszillators im Magnetfeld}, Z.~Physik {\bf 47}, 446--448  (1928).



\bibitem[FLW]{FLW} R. L. Frank, M. Loss, and T. Weidl, {\em P\'olya's
    conjecture in the presence of a constant magnetic field},
  J.\ Eur.\ Math.\ Soc.\ {\bf 11}, 1365--1383, (2009)

\bibitem[FLL]{FLL}
J. Fr\"ohlich, E. H. Lieb, and M. Loss:
{\it Stability of Coulomb systems with magnetic fields. 
I. The one-electron atom. }
Commun.\ Math.\ Phys.\ {\bf 104} 251--270 (1986) 



\bibitem[H]{H} W.~Hughes: {\em An atomic energy bound that gives
    Scott's correction}, Adv.~Math. {\bf 79}, 213--270 (1990).


\bibitem[Iv1]{Iv1} V.I.~Ivrii: {\em Asymptotics of the ground state
    energy of heavy molecules in a strong magnetic field. I. and II.}  Russian
  J.\ Math.\ Phys.\ {\bf 4} (1996), no. 1, 29-–74 ibid.\
{\bf 5}, no. 3, 321–-354 (1998). 

\bibitem[Iv2]{Iv2} V.I.~Ivrii: {\em Heavy molecules in the strong magnetic field.}
Russian J.\ Math.\ Phys.\ {\bf 4} (1996), no. 4, 449–-455. 

\bibitem[Iv3]{Iv3}
  V.I.~Ivrii: {\em Local trace asymptotics in the self-generated
magnetic field.}, arxiv:1108.4188


\bibitem[IS]{IS} V.I.~Ivrii and I.M.~Sigal: {\em Asymptotics of
    the ground state energies of large Coulomb systems}, Ann.~of Math.
  (2), {\bf 138}, 243--335 (1993).




\bibitem[L]{L} E. H. Lieb: {\it Thomas-Fermi and related
theories of atoms and molecules}, Rev. Mod. Phys. {\bf 65}. No. 4, 603-641
(1981)



\bibitem[LL]{LL}  E. H. Lieb, M. Loss: {\em
Stability of Coulomb systems with magnetic fields  II.}
Commun. Math. Phys, {\bf 104} 271--282 (1986)


\bibitem[LLS]{LLS} E. H. Lieb, M. Loss and J. P. Solovej: {\em
Stability of Matter in Magnetic Fields}, Phys. Rev. Lett. {\bf 75},
 985--989 (1995)




\bibitem[LS]{LS} E. H. Lieb and B. Simon: {\em The Thomas-Fermi theory
of atoms, molecules and solids}, Adv. Math. {\bf 23}, 22-116 (1977)

\bibitem[LSY1]{LSY1} E. H. Lieb, J. P. Solovej and J. Yngvason:
{\sl Asymptotics of heavy atoms in high magnetic fields: I. Lowest
Landau band region}, Commun. Pure  Appl. Math. {\bf 47},
513--591 (1994)

\bibitem[LSY2]{LSY2} E. H. Lieb, J. P. Solovej and J. Yngvason:
{\sl Asymptotics of heavy atoms in high magnetic fields: II. Semiclassical
regions. \/}  Commun. Math. Phys. {\bf 161}, 77--124 (1994)

\bibitem[LT]{LT} E.\ H.\ Lieb and W.\ E.\ Thirring, Bound for
  the kinetic energy of fermions which proves the stability of matter,
  {\it Phys.\ Rev.\ Lett.} {\bf 35}, 687--689 (1975). 

\bibitem[LY]{LY}  M. Loss and H.-T. Yau, {\em Stability of Coulomb
systems with magnetic fields: III. Zero energy bound states of
the Pauli operator. \/} Commun. Math. Phys. {\bf 104} (1986), 283-290.



\bibitem[SW1]{SW1} H.~Siedentop and R.~Weikard: {\em On the leading
    energy correction for the statistical model of an atom:
    interacting case}, Commun.~Math.~Phys.~ {\bf 112}, 471--490
  (1987) 

\bibitem[SW2]{SW2} H.~Siedentop and R.~Weikard: {\em On the leading
    correction of the Thomas-Fermi model: lower bound}, Invent.~Math.
  {\bf 97}, 159--193 (1990)

\bibitem[SW3]{SW3} H.~Siedentop and R.~Weikard:
 {\em A new phase space localization technique with
    application to the sum of negative eigenvalues of {S}chr\"odinger
    operators}, Ann.~Sci.~\'Ecole Norm. Sup. (4), {\bf 24}, no.~2,
  215--225 (1991).

\bibitem[S]{S} B. Simon, Functional Integration and Quantum Physics,
  Academic Press, 1979
  
\bibitem[Sob1]{Sob1} A.~V.~Sobolev: 
{\em The quasi-classical asymptotics of local Riesz means for the Schrödinger operator 
in a strong homogeneous magnetic field.}
Duke Math.\ J.\ {\bf 74} (1994), no. 2, 319-–429. 

\bibitem[Sob2]{Sob2} A.~V.~Sobolev: {\em Discrete spectrum asymptotics
 for the Schr\"{o}dinger operator with a singular potential and a magnetic field},
 Rev.~Math.~Phys {\bf 8} (1996) no.~6, 861--903.
  
\bibitem[Sob3]{Sob3} A.~V.~Sobolev: {\em Two-term asymptotics for the
    sum of eigenvalues of the Schrödinger operator with Coulomb
    singularities in a homogeneous magnetic field.} 
    Asymptotic Anal.\ {\bf 13} (1996), no. 4, 393-–421.


\bibitem[SS]{SS}  J. P. Solovej, W. Spitzer:  {\em A new
coherent states approach to semiclassics which gives
Scott's correction.} Comm. Math. Phys.  {\bf 241}  (2003),  no. 2-3, 383--420.

\bibitem[SSS]{SSS}  J. P. Solovej, T.\O.  S\o rensen, W. Spitzer:  {\it 
Relativistic Scott correction for atoms and molecules.}
Comm. Pure Appl. Math. Vol. LXIII. 39-118 (2010).

\bibitem[Y]{Y}
J.~Yngvason: {\em Thomas-Fermi theory for matter in a magnetic field as a limit of quantum mechanics.} 
Lett.\ Math.\ Phys.\ {\bf 22} (1991), no. 2, 107-–117.




\end{thebibliography}
\end{document}